%% file: aaa-blotto-full.tex
\pdfoutput=1

\documentclass[letterpaper,11pt]{article}
\usepackage{amsmath, amssymb}
\usepackage{color}
\usepackage{fullpage}
\usepackage{url}
\usepackage[numbers]{natbib}
\usepackage[noend]{algorithmic}
\usepackage{tikz}
\usepackage{enumerate}
\usepackage{hyperref}
\usepackage{graphicx}
\usepackage{caption}
\usepackage{subcaption}
\usepackage{zerosum,csquotes}
\usepackage{enumitem,linegoal}
\usepackage[linesnumbered,ruled,vlined]{algorithm2e}
\usepackage{comment}	
\usepackage{ctable}

\usepackage{mathtools}
\usepackage[flushleft]{threeparttable}
\usepackage{empheq}
\usepackage{newfloat}

\usepackage{footnote}
\makesavenoteenv{tabular}
\makesavenoteenv{table}

\makeatletter
\g@addto@macro{\UrlBreaks}{\UrlOrds}
\makeatother

\usepackage[margin=1in]{geometry}

 \newtheorem{theorem}{Theorem}[section]
 \newtheorem{lemma}[theorem]{Lemma}
 
 \newtheorem{corollary}[theorem]{Corollary}
 
 \newtheorem{definition}[theorem]{Definition}

\makeatletter
\def\GrabProofArgument[#1]{ #1: \egroup\ignorespaces}
\def\proof{\noindent\textbf\bgroup Proof%
	\@ifnextchar[{\GrabProofArgument}{. \egroup\ignorespaces}}

\makeatother

\input{macros.tex}

\newcommand{\footremember}[2]{%
   \footnote{#2}
    \newcounter{#1}
    \setcounter{#1}{\value{footnote}}%
}
\newcommand{\footrecall}[1]{%
 \footnotemark[\value{#1}]
}

\newcounter{proccnt}

\title{Faster and Simpler Algorithm for Optimal Strategies of \\Blotto Game}

\author{
	Soheil Behnezhad\footremember{nsfdarpa}{Supported in part by NSF CAREER award CCF-1053605, NSF BIGDATA grant IIS-1546108, NSF AF:Medium grant CCF-1161365, DARPA GRAPHS/AFOSR grant FA9550-12-1-0423, and another DARPA SIMPLEX grant.},
	Sina Dehghani\footrecall{nsfdarpa},
	Mahsa Derakhshan\footrecall{nsfdarpa}, \\
	MohammadTaghi HajiAghayi\footrecall{nsfdarpa},
	Saeed Seddighin\footrecall{nsfdarpa}
	\\\\
Department of Computer Science, University of Maryland\\
\texttt{\{soheil, dehghani, mahsaa, hajiagha, sseddigh\}@cs.umd.edu
}}

\begin{document}
\renewcommand{\theenumi}{(\roman{enumi}).}
\renewcommand{\labelenumi}{\theenumi}
\sloppy

%
%

\date{}

\maketitle

\begin{abstract}
\input{abstract.tex}
\end{abstract}

\input{intro.tex}

\input{preliminaries.tex}

\input{blotto.tex}

\input{lp.tex}

\input{lowerbound.tex}

\input{multidimen.tex}
\input{experimental.tex}

\clearpage

\bibliographystyle{aaai}
\bibliography{blotto}


\end{document}

%% file: macros.tex
\newcommand{\ahmadinejadduels}{Ahmadinejad \textit{et al.}}
\newcommand{\yannikakisexpressing}{Yannakakis}
\newcommand{\martinusing}{Martin}
\newcommand{\payoffa}{U^A}
\newcommand{\payoffb}{U^B}
\newcommand{\troopsA}{A}
\newcommand{\troopsB}{B}
\newcommand{\flow}{F}

\newcommand{\battlefields}{\ensuremath{K}}
\newcommand{\troops}{\ensuremath{N}}
\newcommand{\exactlowerbound}{\ensuremath{\frac{\battlefields(\troops - \lceil \frac{\troops+1}{2}\rceil)(\troops - \lceil \frac{\troops+1}{2}\rceil + 1)}{2}}}
\newcommand{\orderlowerbound}{\ensuremath{\Theta(\troops^2\battlefields)}}

\newcommand{\stratmappinga}{\ensuremath{\mathcal{G}^A}}
\newcommand{\stratmappingb}{\ensuremath{\mathcal{G}^B}}
\newcommand{\stratmappingaMRCB}{\ensuremath{\stratmappinga_{\mathcal{M}}}}
\newcommand{\stratmappingbMRCB}{\ensuremath{\stratmappingb_{\mathcal{M}}}}

\newcommand{\mainpolytopea}{\ensuremath{\mathcal{R}^A}}
\newcommand{\mainpolytopeb}{\ensuremath{\mathcal{R}^B}}
\newcommand{\mainpolytopeaMRCB}{\ensuremath{\mainpolytopea_{\mathcal{M}}}}
\newcommand{\mainpolytopebMRCB}{\ensuremath{\mainpolytopeb_{\mathcal{M}}}}

\newcommand{\polytopea}{\ensuremath{\mathcal{P}_A}}
\newcommand{\polytopeb}{\ensuremath{\mathcal{P}_B}}
\newcommand{\polytopeaMRCB}{\ensuremath{\polytopea^{\mathcal{M}}}}
\newcommand{\polytopebMRCB}{\ensuremath{\polytopeb^{\mathcal{M}}}}

\newcommand{\xc}[1]{\operatorname{xc}(#1)}
\newcommand{\slack}[1]{S^{#1}}
\newcommand{\rk}[0]{\ensuremath{\operatorname{rk}_+}}

\newcommand{\resources}[0]{\ensuremath{c}}
\newcommand{\xcMRCB}{\troops^{2\resources}\battlefields}
\newcommand{\orderxcMRCB}{\Theta(\xcMRCB)}
\newcommand{\variablesMRCB}{\ensuremath{\Theta(\troops^{2\resources}\battlefields)}}

\newcommand{\variables}{\ensuremath{\Theta(\troops^2\battlefields)}}
\newcommand{\gpayoff}{\ensuremath{u}} 
\newcommand{\layeredgraph}[1]{\ensuremath{\mathcal{L}^{#1}}}

\makeatletter
\def\env@cases{%
  \let\@ifnextchar\new@ifnextchar
  \left.
  \def\arraystretch{1.2}%
  \array{@{}l@{\,\,}l@{}}%
}%
\makeatother
\newcommand{\pr}[2]{\ensuremath{P^{A}_{#1, #2}}}
\newcommand{\weight}[2]{\ensuremath{W^{B}_{#1, #2}}}
\newcommand{\lpvard}[2]{\ensuremath{D^B_{#1, #2}}}

\newcommand{\pipe}{\ensuremath{\mathrel{}|\mathrel{}}}
\DeclareFloatingEnvironment[fileext=lop, name=Linear Program]{lp}


%% file: abstract.tex
In the Colonel Blotto game, which was initially introduced by Borel in 1921, two colonels simultaneously distribute their troops across different battlefields. The winner of each battlefield is determined independently by a winner-take-all rule. The ultimate payoff of each colonel is the number of battlefields he wins. The Colonel Blotto game is commonly used for analyzing a wide range of applications from the U.S presidential election, to innovative technology competitions, to advertisement, to sports, and to politics. There has been persistent efforts for finding the optimal strategies for the Colonel Blotto game. After almost a century Ahmadinejad, Dehghani, Hajiaghayi, Lucier, Mahini, and Seddighin~\cite{ahmadinejad2016duels} provided an algorithm for finding the optimal strategies in polynomial time. 

\ahmadinejadduels~\cite{ahmadinejad2016duels} first model the problem by a Linear Program (LP) with both an exponential number of variables and an exponential number of constraints which makes the problem intractable. Then they project their solution to another space to obtain another exponential-size LP, for which they can use Ellipsoid method. However, despite the theoretical importance of their algorithm, it is highly impractical. In general, even Simplex method (despite its exponential running time in practice) performs better than Ellipsoid method in practice.

In this paper, we provide the first polynomial-size LP formulation of the optimal strategies for the Colonel Blotto game. We use linear extension techniques. Roughly speaking, we project the strategy space polytope to a higher dimensional space, which results in lower number of facets for the polytope. In other words, we add a few variables to the LP, such that surprisingly the number of constraints drops down to a polynomial. We use this polynomial-size LP to provide a novel simpler and significantly faster algorithm for finding optimal strategies for the Colonel Blotto game.

We further show this representation is asymptotically tight, which means there exists no other linear representation of the problem with less number of constraints. We also extend our approach to multi-dimensional Colonel Blotto games, where each player may have different sorts of budgets, such as money, time, human resources, etc.

By implementing this algorithm we were able to run tests which were previously impossible to solve in a reasonable time. These informations, allow us to observe some interesting properties of Colonel Blotto; for example we find out the behaviour of players in the discrete model is very similar to the continuous model Roberson~\cite{R06} solved. 

%% file: intro.tex
\section{Introduction}
In the U.S. presidential election, the President is elected by the Electoral College system. In the Electoral College system, each state has a number of electoral votes, and the candidate who receives the majority of electoral votes is elected as the President of the United States. In most of the states\footnote{All states except Maine and  Nebraska},
a winner-take-all role determines the electoral votes, and the candidate who gets the majority of votes in a state will benefit from all the electoral votes of the corresponding state. 
Since the President is not elected by national popular vote directly, any investment in the states which are highly biased toward a party would be wasted. For example, a Democratic candidate can count on the electoral votes of states like California, Massachusetts, and New York, and a Republican candidate can count on the electoral votes of states like Texas, Mississippi, and South Carolina. This highlights the importance of those states that are likely to choose either parties, and would determine the outcome of the election. These states, known as {\em swing states} or {\em battleground states,} are the main targets of a campaign during the election, e.g., the main battleground states of the 2012 U.S. presidential election were Colorado, Florida, Iowa, New Hampshire, North Carolina, Ohio, Virginia, and Wisconsin. Now answers to the following questions seem to be essential: how can a national campaign distribute its resources like time, people, and money across different battleground states? What is the outcome of the game between two parties?

One might see the same type of competition between two companies which are developing new technologies.
These companies need to distribute their efforts across different markets. 
The winner of each market would become the market-leader and takes almost all the benefits of the corresponding market \cite{KR10,KR12}.  
For instance, consider the competition between Samsung and Apple, where they both invest on developing products like cell-phones, tablets, and laptops, and all can  have different specifications. Each product has its own specific market and the most plausible brand will lead that market. Again, a strategic planner with limited resources would face a similar question: what would be the best strategy for allocating the resources across different markets? 

\noindent \textbf{Colonel Blotto Game.}
The {\em Colonel Blotto} game, which was first introduced by Borel \cite{B21}, provides a model to study the aforementioned problems. This paper was later discussed in an issue of {\em Econometria} \cite{B53,F53a,F53b,V53}. 
Although the Colonel Blotto model was initially proposed to study a war situation, it has been  applied for analyzing the competition in different contexts from sports, to advertisement, and to politics \cite{M93,LP02,MMT05,CKS09,KR10,KR12}.  
In the original Colonel Blotto game two colonels fight against each other over different battlefields. They should simultaneously divide their troops among different battlefields without knowing the actions of their opponents.
A colonel wins a battlefield if and only if the number of his troops dominates the number of troops of his opponent. 
The final payoff of each colonel, in its classical form, is the number of the battlefields he wins. 
The {\em MaxMin} strategy of a player maximizes the minimum gain that can be achieved. In two player zero-sum games a MaxMin strategy is also the optimal strategy, since any other strategy may result in a lower payoff across a rational player. It is also worth mentioning that in zero-sum games a pair of strategies is a Nash equilibrium if and only if both players are playing MaxMin strategies. Therefore finding MaxMin strategies results in finding the optimal strategies for players and also the Nash equilibria of the game. It is easy to show that solving Colonel Blotto is computationally harder than finding the optimal strategies of any two player game. Consider a two player game in which the players have $x$ and $y$ pure strategies, respectively. Now, construct a Blotto game with two battlefield in which player A has $x-1$ troops and player B has $y-1$ troops. Note that in this  Blotto game, the number of strategies of the players are $x$ and $y$ respectively, and one can easily encode the payoff function of the original game in the partial payoffs of the two battlefields. Thus, any solution for Colonel Blotto yields a solution for an arbitrary two player game.

Colonel Blotto is a zero-sum game, but the fact that the number of pure strategies of the agents are exponential in the number of troops and the number of battlefields, makes the problem of finding optimal strategies quite hard. There were several attempts for solving variants of the problem  since 1921  \cite{T49,Bl54,Bl58,Be69,SW81,W05,R06,K07,H07,GP09,KR12}.
Most of the works consider special cases of the problem. For example many results in the literature relax the integer constraint of the problem, and study a {\em continuous} version of the problem where troops are divisible. 
For example, Borel and Ville~\cite{BV38} 
proposed the first solution for three battlefields. Gross and Wagner~\cite{GW50} generalized this result for any number of battlefields. However, they assumed colonels have the same number of troops.
%
Roberson~\cite{R06} computes the optimal strategies of the Blotto games in the continuous version of the problem where all the battlefields have the same weight, i.e. the game is symmetric across the battlefields. Hart~\cite{H07} considered the discrete version, again when the game is symmetric across the battlefields, and solved it for some special cases. Very recently Ahmadinejad, Dehghani, Hajiaghayi, Lucier, Mahini, and Seddighin~\cite{ahmadinejad2016duels} made a breakthrough in the study of this problem by finding optimal strategies for the Blotto games after nearly a century, which brought a lot of attention \cite{cover1,cover2,cover3,cover4,cover5,cover6,cover7,cover8,cover9}. They obtain exponential sized LPs, and then provide a clever use of Ellipsoid method for finding the optimal strategies in polynomial time.

Although theoretically Ellipsoid method is a very powerful tool with deep consequences in complexity and optimization, it is ``too inefficient to be used in practice" \cite{bernhard}. Interior point methods, and Simplex method (even though it has exponential running-time in the worst case) are "far more efficient" \cite{bernhard}. Thus a practical algorithm for finding optimal strategies for the Blotto games remains an open problem. In fact there has been huge studies in existence of efficient LP reformulations for different exponential-size LPs. For example Rothvoss~\cite{rothvoss2014matching} proved that the answer of the long-standing open problem, asking whether a graph's perfect matching polytope can be represented by an LP with polynomial number of constraints, is negative. The seminal work of Applegate and Cohen~\cite{applegate2003making} also provides polynomial-size LPs for finding an optimal oblivious routing.
We are the first to provide a polynomial-size LP for finding the optimal strategies of the Colonel Blotto games. Although \ahmadinejadduels~\cite{ahmadinejad2016duels} use an LP with exponential number of constraints, our LP formulation has only $O(\troops^2\battlefields)$ constraints, where $\troops$ denotes the number of troops and $\battlefields$ denotes the number of battlefields. Consequently we provide a novel simpler and significantly faster algorithm using the polynomial-size LP.

Furthermore we show that our LP representation is asymptotically tight. The rough idea behind obtaining a polynomial-size LP is the following. Given a polytope $P$ with exponentially many facets, we project $P$ to another polytope $Q$ in a higher dimensional space which has polynomial number of facets. Thus basically we are adding a few variables to the LP in order to reduce the number of constraints down to a polynomial.
$Q$ is called the {\em linear extension} of $P$. The minimum number of facets of any linear extension is called the {\em extension complexity}. We show that the extension complexity of the polytope of the optimal strategies of the Colonel Blotto game is \orderlowerbound{}. In other words, there exists no LP-formulation for the polytope of MaxMin strategies of the Colonel Blotto game with fewer than \orderlowerbound{} constraints.

We also extend our approach to the {\em Multi-Resource Colonel Blotto} (MRCB) game. In MRCB, each player has different types of resources. Again the players distribute their budgets in the battlefields. Thus each player allocates a vector of resources to each battlefield. The outcome in each battlefield is a function of both players' resource vectors that they have allocated to that battlefield. MRCB models a very natural and realistic generalization of the Colonel Blotto game. For example in U.S. presidential election, the campaigns distribute different resources like people, time, and money among different states. We provide an LP formulation for finding optimal strategies in MRCB with $\orderxcMRCB{}$ constraints and $\variablesMRCB{}$ variables, where $c$ is the number of resources. We prove this result is also tight up to constant factors, since the extension complexity of MRCB is $\orderxcMRCB$.

By implementing our LP, we observe the payoff of players in the continuous version considered by Roberson~\cite{R06} very well predicts the outcome of the game in the auctionary and symmetric version of our model.

%% file: preliminaries.tex
\section{Preliminaries}
Throughout this paper we assume the number of battlefields is denoted by $\battlefields$ and the number of troops of players A and B are denoted by $\troopsA$ and $\troopsB$ respectively. Also in some cases we use $\troops$ to denote the number of troops of an unknown player.

Generally mixed strategies are shown by a probability vector over pure strategies. However at some points in this paper we project this representation to another space that specifies probabilities to each battlefield and troop count pair. More precisely, we map a mixed strategy $x$ of player A to $\stratmappinga(x) = \hat{x} \in \lbrack 0, 1\rbrack^{d(A)}$ where $d(A) = \battlefields\times(\troopsA + 1)$. We may abuse this notation for convenience and use $\hat{x}_{i, j}$ to show the probability the mixed strategy $x$ puts $j$ troops in the $i$-th battlefield. Note that this mapping is not one-to-one. Similarly, we define $\stratmappingb(x)$ to map a mixed strategy $x$ of player B to a point in $\lbrack 0, 1 \rbrack^{d(B)}$ where $d(B) = \battlefields\times(\troopsB + 1)$. Let \mainpolytopea{} and \mainpolytopeb{} denote the set of all possible mixed strategies of A and B in a Nash equilibrium. We define $\polytopea = \{\hat{x}\pipe\exists x \in \mainpolytopea, \stratmappinga(x)=\hat{x}\}$ and $\polytopeb = \{\hat{x}\pipe\exists x \in \mainpolytopeb, \stratmappingb(x)=\hat{x}\}$ to be the set of all Nash equilibrium strategies in the new space for A and B respectively.

Multi-Resource Colonel Blotto is a generalization of Colonel Blotto where each player may have different types of resources. In MRCB, there are $\battlefields{}$ battlefields and $\resources$ resource types. Players simultaneously distribute all their resources of all types over the battlefields. Let $\troopsA_i$ and $\troopsB_i$ denote the number of resources of type $i$ player A and B respectively have. A pure strategy of a player would be a partition of his resources over battlefields.  In other words, let $x_{i, j}$ and $y_{i, j}$ denote the amount of resources of type $j$, players A and B put in battlefield $i$ respectively. A vector $x = \langle x_{1, 1}, \ldots, x_{\battlefields, \resources} \rangle$ is a pure strategy of player A if for any $1 \leq j \leq \resources$, $\sum_{i=1}^{\battlefields} x_{i, j} = \troopsA_j$. Similarly a vector $y = \langle y_{1, 1}, \ldots, y_{\battlefields, \resources} \rangle$ is a pure strategy of player B if for any $1 \leq j \leq \resources$, $\sum_{i=1}^{\battlefields} y_{i, j} = \troopsB_j$. Let $\payoffa(x, y)$ and $\payoffb(x, y)$ denote the payoff of A and B and let $\payoffa_i(x, y)$ and $\payoffb_i(x, y)$ show their payoff over the $i$-th battlefield respectively. Note that 
$$\payoffa(x, y) = \sum_{i=1}^{\battlefields}\payoffa_i(x, y)$$ and 
$$\payoffb(x, y) = \sum_{i=1}^{\battlefields}\payoffb_i(x, y).$$ On the other hand since MRCB is a zero-sum game $\payoffa_i(x, y) = -\payoffb_i(x, y)$. Similar to Colonel Blotto we define $\mainpolytopeaMRCB{}$ and $\mainpolytopebMRCB{}$ to denote the set of all possible mixed strategies of A and B in a Nash equilibrium of MRCB and for any mixed strategy x for player A we define the mapping $\stratmappingaMRCB(x) = \hat{x} \in \lbrack 0, 1\rbrack^{d^\mathcal{M}(A)}$ where $d^\mathcal{M}(A) = \battlefields\times(\troopsA_1 + 1)\ldots\times(\troopsA_\resources + 1p)$ and by $\hat{x}_{i, j_1, \ldots, j_\resources}$ we mean the probability that in mixed strategy $x$, A puts $j_t$ amount of resource type $t$ in the $i$-th battlefield for any $t$ where $1 \leq t \leq \resources$.
We also define the same mapping for player B, $\stratmappingbMRCB(x) = \hat{x} \in \lbrack 0, 1\rbrack^{d^\mathcal{M}(B)}$ where $d^\mathcal{M}(B) = \battlefields\times(\troopsB_1 + 1)\ldots\times(\troopsB_\resources + 1)$. Lastly we define $\polytopeaMRCB = \{\hat{x}\pipe\exists x \in \mainpolytopeaMRCB, \stratmappingaMRCB(x)=\hat{x}\}$ and $\polytopebMRCB = \{\hat{x}\pipe\exists x \in \mainpolytopebMRCB, \stratmappingbMRCB(x)=\hat{x}\}$ to be the set of all Nash equilibrium strategies after the mapping.

%% file: blotto.tex
\section{LP Formulation}
In this section we explain the LP formulation of Colonel Blotto proposed by \ahmadinejadduels~\cite{ahmadinejad2016duels} and show how it can be reformulated in a more efficient way. Recall that in the Colonel Blotto game, we have two players A and B, each in charge of a number of troops, namely $\troopsA$ and $\troopsB$ respectively. Moreover, the game is played on $\battlefields$ battlefields and every player's pure strategy is an allocation of his troops to the battlefields. Therefore, the number of pure strategies of the players is $\binom{\troopsA+\battlefields-1}{\battlefields-1}$ for player A and $\binom{\troopsB+\battlefields-1}{\battlefields-1}$ for player B.

The conventional approach to formulate the mixed strategies of a game is to represent every strategy by a vector of probabilities over the pure strategies. More precisely, a mixed strategy of a player is denoted by a vector of size equal to the number of his pure strategies, whose every element indicates the likelihood of taking a specific action in the game. The only constraint that this vector adheres to, is that the probabilities are non-negative and add up to 1. Such a formulation for Colonel Blotto requires a huge amount of space and computation, since the number of pure strategies of each player in this game is exponentially large.

To overcome this hardness, \ahmadinejadduels ~\cite{ahmadinejad2016duels} propose a more concise representation that doesn't suffer from the above problem. This is of course made possible by taking a significant hit on the simplicity of the description. They suggest, instead of indicating the probability of taking every action in the representation, we only keep track of the probabilities that a mixed strategy allocates a certain amount of troops to every battlefield. In other words, in the new representation, for every number of troops and any battlefield we have a real number, that denotes the probability of allocating that amount of troops to the battlefield. As a result, the length of the representation reduces from the number of pure strategies to $(\troopsA+1)\battlefields$ for player A and $(\troopsB+1)\battlefields$ for player B. This is indeed followed by a key observation: given the corresponding representations of the strategies of both players, one can determine the outcome of the game regardless of the actual strategies. In other words, the information stored in the representations of the strategies suffices to determine the outcome of the game.

In contrast to the conventional formulation, \ahmadinejadduels's representation is much more complicated and not well-understood. For example, in order to see if a representation corresponds to an actual strategy in the conventional formulation, we only need to verify that all of the probabilities are non-negative and their total sum is equal to 1. \ahmadinejadduels's representation, however, is not trivial to verify. Apart from the trivial constraints such as the probabilities add up to 1 or the number of allocated troops matches the number of the player's troops, there are many other constraints to be met. Moreover, it is not even proven whether such a representation can be verified with a polynomial number of linear constraints.

\ahmadinejadduels ~\cite{ahmadinejad2016duels} leverage the new representation to determine the equilibria of Colonel Blotto in polynomial time. Recall that in zero-sum games such as Colonel Blotto, the minmax strategies are the same as the maxmin strategies, and the game is in Nash Equilibrium if and only if both players play a maxmin strategy~\cite{nisan2007algorithmic}. Roughly speaking, the high-level idea of \ahmadinejadduels\enspace is to find a mixed strategy which performs the best against every strategy of the opponent. By the equivalence of the minmax and maxmin strategies then, one can show such a strategy is optimal for that player. Therefore, the naive formulation of the equilibria of Blotto is as follows:


\begin{alignat}{3}
	\label{lp1}
	\max  \quad & \gpayoff & \\
    \text{s.t. }\qquad
    & \hat{x} \text{ is a valid strategy for player A}\nonumber\\
    & \payoffa(\hat{x},\hat{y}) \geq \gpayoff \qquad \forall \hat{y}\nonumber
\end{alignat}

Note that, $\hat{x}$ is a vector of size $(\troopsA+1)\battlefields$ that represents a strategy of player A. Similarly, for every mixed strategy of player B, represented by $\hat{y}$, we have a constraint to ensure $\hat{x}$ achieves a payoff of at least $\gpayoff$ against $\hat{y}$. Notice that the only variables of the program are the probabilities encoded in vector $\hat{x}$. All other parameters are given as input, and hence appear as constant coefficients in the program. As declared, there are two types of constraints in Program \ref{lp1}. The first set of constraints ensures the validity of $\hat{x}$, and the second set of constraints makes sure $\hat{x}$ performs well against every strategy of player B. \ahmadinejadduels ~\cite{ahmadinejad2016duels} call the first set \textit{the membership constraints} and the second set \textit{the payoff constraints}. Throughout the paper Since for every mixed strategy, there exists a best response of the opponent which is pure, one can narrow dawn the payoff constraints to the pure strategies of player B.

The last observation of \ahmadinejadduels ~\cite{ahmadinejad2016duels} is to show both types of the constraints are convex in the sense that if two strategy profiles $\hat{x_1}$ and $\hat{x_2}$ meet either set of constraints, then $\frac{\hat{x_1}+\hat{x_2}}{2}$ is also a feasible solution for that set. This implies that Program \ref{lp1} is indeed a linear program that can be solved efficiently via the Ellipsoid method. However, \ahmadinejadduels's algorithm is practically impossible to run, as its computational complexity is $O(\troops^{12}\battlefields^4)$.

The reason \ahmadinejadduels's algorithm is so slow is that their LP has exponentially many constraints. Therefore, they need to run the Ellipsoid algorithm run solve the program. In addition to this, their separation oracle is itself a linear program with exponentially many constraints which is again very time consuming to run. However, a careful analysis shows that these exponentially many constraints are all necessary and none of them are redundant. This implies that the space of the LP as described by \ahmadinejadduels\enspace requires exponentially many constraints to formulate and hence we cannot hope for a better algorithm. A natural question that emerges, however, is whether we can change the space of the LP to solve it with a more efficient algorithm?

In this paper we answer the above question in the affirmative. There has been persistent effort to find efficient formulations for many classic polytopes. As an example, \textit{spanning trees} of a graph can be formulated via a linear program that has an exponential number of linear constraints. It is also not hard to show none of those constraints are redundant~\cite{edmonds1971matroids}. However, \martinusing~\cite{martin1991using} showed that the same polytope can be formulated with $O(n^3)$ linear constraints where $n$ is the number of nodes of the graph. Other examples are \textit{the permutahedron} ~\cite{goemans2015smallest}, \textit{the parity polytope}~\cite{rothvoss2013some}, and \textit{the matching polytope}~\cite{rothvoss2014matching}. In these examples, a substantial decrease in the number of constraints of the linear formulation of a problem is made possible by adding auxiliary variables to the program. Our work follows the same guideline to formulate the equilibria of Blotto with a small number of constraints.

In Section \ref{lp}, we explain how to formulate the membership and payoff limitations with a small number of linear constraints. Finally in Section \ref{lowerbound}, we show that our formulation is near optimal. In other words, we show that any linear program that formulates the equilibria of Blotto, has to have as many linear constraints as the number of constraints in our formulation within a constant factor. We show this via \textit{rectangle covering lower bound} proposed by \yannikakisexpressing~\cite{yannakakis1988expressing}

%% file: lp.tex
\section{Main Results}\label{lp}
In this section we give a linear program to find a maxmin strategy for a player in an instance of Colonel Blotto with polynomially many constraints and variables. To do this, we describe the same representation described by \ahmadinejadduels{}'s \cite{ahmadinejad2016duels} LP in another dimension, to reduce the number of constraints. This gives us a much better running time, since they had to use ellipsoid method to find a solution for their LP in polynomial time, which makes their algorithm very slow and impractical. We define a \textit{layered graph} for each player and show any mixed strategy of a player can be mapped to a particular flow in his layered graph. Our LP includes two set of constraints, \textit{membership constraints} and \textit{payoff constraints}. Membership constraints guarantee we find a valid strategy and payoff constraints guarantee this strategy minimizes the maximum benefit of the other player.

\begin{definition}[Layered Graph]\label{layeredGraph}
For an instance of a Blotto game with \battlefields{} battlefields, we define a layered graph for a player with $\troops$ troops as follows: The layered graph has $\battlefields{}+1$ layers and $\troops{}+1$ vertices in each layer. Let $v_{i,j}$ denote the $j$'th  vertex in the $i$'th layer ( $0\leq i\leq \battlefields{}$ and $0\leq j\leq \troops{}$). For any $1 \leq i \leq \battlefields{}$ there exists a directed edge from $v_{i-1, j}$ to $v_{i, l}$ iff $0 \leq j\leq l \leq \troops{}$. We denote the layered graph of player A and B by \layeredgraph{A} and \layeredgraph{B} respectively.
\end{definition}
Based on the definition of layered graph we define \textit{canonical paths} as follows:
\begin{definition}[Canonical Path]\label{canonicalPath}
A canonical path is a directed path in a layered graph that starts from $v_{0,0}$ and ends at $v_{\battlefields, \troops}$.
\end{definition}
\begin{figure}[hbt]
	\centering
	\includegraphics[scale=0.7]{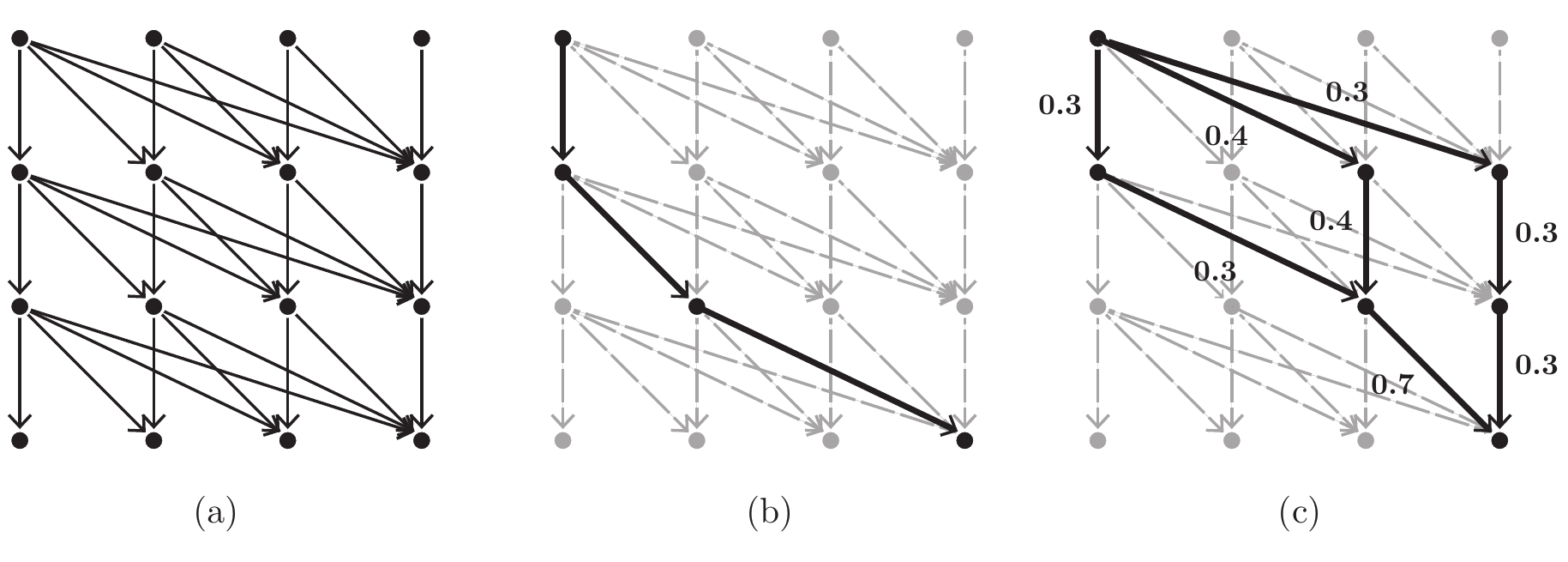}
	\caption{Figure (a) shows a layered graph for a player with $3$ troops playing over $3$ battlefields. In Figure (b) a canonical path corresponding to a pure strategy where the player puts no troops on the first battlefield, 1 troop on the second one and two troops on the 3rd one is shown. Figure (c) shows a flow of size $1$, that is a representation of a mixed strategy consisting of three pure strategies with probabilities $0.3$, $0.4$ and $0.3$.}
		\label{fig:layGraphCanPath}
\end{figure}
Figure~\ref{fig:layGraphCanPath} shows a layered graph and a canonical path. Now, we give a one-to-one mapping between canonical paths and pure strategies.

\begin{lemma}\label{pureLemma}
	Each pure strategy for a player is equivalent to exactly one canonical path in the layered graph of him and vice versa.
\end{lemma}

\begin{proof}
	Since the edges in the layered graph exist only between two consecutive layers, each canonical path contains exactly $\battlefields$ edges. Let $p$ be an arbitrary canonical path in the layered graph of a player with $\troops{}$ troops. In the equivalent pure strategy put $l_i$ troops in the battlefield $i$ if $p$ contains the edge between $v_{i-1, j}$ and $v_{i, j+l_i}$ for some $j$. By definition of the layered graph, we have  $l_i\geq 0$. Also since $p$ starts from $v_{0,0}$ and ends in $v_{\battlefields, \troops}$ we have $\sum_{i=0}^{\battlefields} l_i= \troops $. Therefore this strategy is a valid pure strategy.
			
	On the other hand, let $s$ be a valid pure strategy and let $s_i$ denote the total number of troops in battlefields $1$ to $i$ in strategy $s$. We claim the set of edges between $v_{i-1, s_{i-1}}$ and $v_{i, s_i}$ for $1 \leq i \leq \battlefields$ is a canonical path. Note that $s_0 = 0$ and $s_{\battlefields}=\troops$ also the endpoint of any of such edges is the starting point of the edge chosen from the next layer, so we have constructed a valid canonical path.
\end{proof}

So far it is clear how layered graphs are related to pure strategies using canonical paths. Now we explain the relation between mixed strategies and flows of size $1$ where $v_{0, 0}$ is the source and $v_{\battlefields, \troops}$ is the sink. One approach to formulate the mixed strategies of a game is to represent every strategy by a vector of probabilities over the pure strategies. Since based on Lemma \ref{pureLemma} each pure strategy is equivalent to a canonical path in the layered graph; for any pure strategy $s$ with probability $P(s)$ in a mixed strategy we assign a flow of size $P(s)$ to the corresponding canonical paths of $s$ in the layered graph. All these paths begin and end in $v_{0,0}$ and $v_{\battlefields, \troops}$ respectively. Therefore since $\sum P(s) = 1$ for all pure strategies of a mixed strategy, the size of the corresponding flow would be exactly $1$.
\begin{corollary}\label{mixtoflowLemma}
	For any mixed strategy of a player with \troops{} troops there is exactly one corresponding flow from vertex $v_{0,0}$ to $v_{\battlefields,\troops}$ in the layered graph of that player.
\end{corollary}
Note that although we map any given mixed strategy to a flow of size $1$ in the layered graph, this is not a one-to-one mapping because several mixed strategies could be mapped to the same flow. However in the following lemma we show that this mapping is surjective.
\begin{lemma}\label{flowtomixLemma}
	For any flow of size $1$ from $v_{0,0}$ to $v_{\battlefields, \troops}$ in the layered graph of a player with \troops{} troops, there is at least one mixed strategy of that player with a polynomial size support that is mapped to this flow.
\end{lemma}

\begin{proof}
First, note that we can decompose any given flow to polynomially many flow paths from source to sink \cite{cormen2009introduction}. A flow path is a flow over only one path from source to sink. One algorithm to find such decomposition finds a path $p$ from source to sink in each step and subtracts the minimum passing flow through its edges from every edge in $p$. The steps are repeated until there is no flow from source to sink. Since the flow passing through at least one edge becomes $0$ at each step, the total number of these paths will not exceed the total number of edges in the graph. This means the number of flow paths in the decomposition will be polynomial.

Now, given a flow of size $1$ from $v_{0, 0}$ to $v_{\battlefields, \troops}$, we can basically decompose it to polynomially many flow paths using the aforementioned algorithm. The paths over which these flow paths are defined correspond to pure strategies and the amount of flow passing through each, corresponds to its probability in the mixed strategy.
\end{proof}

Using the flow representation for mixed strategies and the shown properties for it, we give the first LP with polynomially many constraints and variables to find a maxmin strategy for any player in an instance of Colonel Blotto. Our LP consists of two set of constraints, the first set (membership constraints) ensures we have a valid flow of size $1$. This means we will be able to map the solution to a valid mixed strategy. The second set of constraints are needed to ensure the minimum payoff of the player we are finding the maxmin strategy for, is at least \gpayoff{}. Now, by maximizing \gpayoff{} we will get a maxmin strategy. In the following theorem we prove $\polytopea$ could be formulated with polynomially many constraints and variables. Note that one can swap $\troopsA$ and $\troopsB$ and use the same LP to formulate $\polytopeb{}$.

\begin{theorem}
	In an instance of Colonel Blotto, with \battlefields{} battlefields and at most \troops{} troops for each player, \polytopea{} could be formulated with \orderlowerbound{} constraints and \variables{} variables.
\end{theorem}\label{lpMRCB}
\begin{proof}
The high-level representation of our LP is as follows:
\begin{alignat}{3}\label{highlevellp}
	\max  \quad & \gpayoff & \\
    \text{s.t. }\qquad
    & \hat{x} \text{ is a valid strategy for player A}\nonumber\\
    & \payoffb(\hat{x},\hat{y}) \leq -\gpayoff \qquad \forall \hat{y}.\nonumber
\end{alignat}
The strategies $\hat{x}$ and $\hat{y}$ are represented using a flow of size $1$ in the layered graph of player A and B respectively. In Lemma \ref{flowtomixLemma} we proved any valid flow representation could be mapped to a mixed strategy.

To ensure we a have a valid flow of size $1$ from $v_{0, 0}$ to $v_{\battlefields, \troopsA}$ in \layeredgraph{A} (recall that \layeredgraph{A} denotes the layered graph of player A), we use the classic LP representation of flow \cite{bazaraa2011linear}. That is, not having any negative flow and the total incoming flow of each vertex must be equal to its total outgoing flow except for the source and the sink. We denote the amount of flow passing through the edge from $v_{k, i}$ to $v_{k+1, j}$ by variable $\flow_{k, i, j}$. The exact membership constraints are shown in Linear Program~\ref{fig:detailedlp}-a.

\begin{figure}
	\centering
	\includegraphics[scale=0.8]{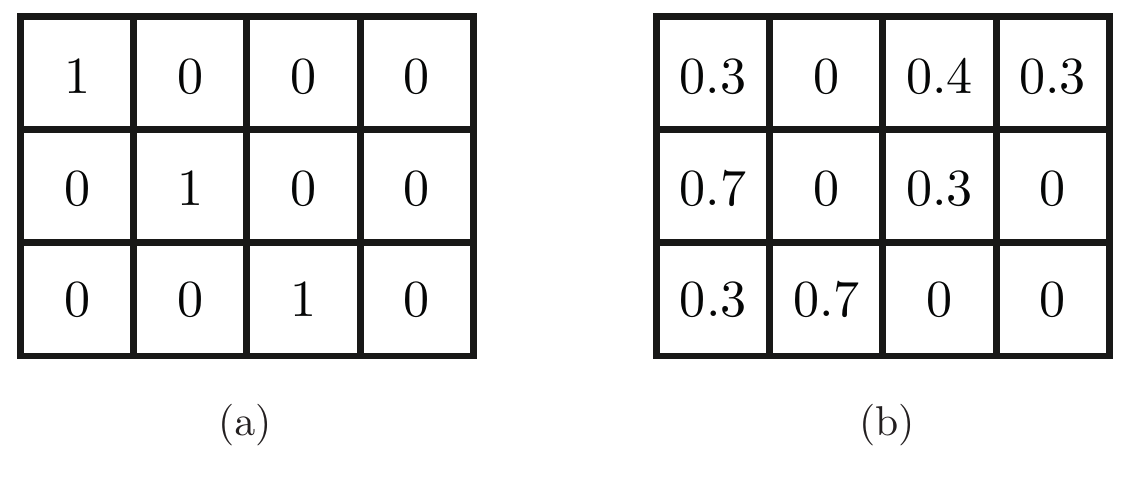}
	\caption{Figure (a) shows $\pr{k}{i}$ for the pure strategy specified in Figure~\ref{fig:layGraphCanPath}-b and Figure (b) shows $\pr{k}{i}$ for the mixed strategy specified in Figure~\ref{fig:layGraphCanPath}-c. The rows correspond to battlefields and the columns correspond to the number of troops.}
	\label{fig:prob}
\end{figure}
On the other hand, we maximize the guaranteed payoff of player A, by bounding the maximum possible payoff of player B. To do this, first note that for any given strategy of player A, there exists a pure strategy for player B, that maximizes his payoff. Let $\pr{k}{j}$ denote the probability that player A puts $j$ troops in the $k$-th battlefield. Figure~\ref{fig:prob} shows the value of $\pr{k}{j}$ for the illustrated examples in Figure~\ref{fig:layGraphCanPath}. We can compute these probabilities using the variables defined in the previous constraints, as follows:
\begin{equation}
	\pr{k}{i}= \sum_{i=0}^{\troopsA-j} \flow_{k,i,i+j}
\end{equation}
By having these probabilities we can compute the expected payoff that player B gets over battlefield $k$, if he puts $i$ troops in it. Moreover consider a given canonical path $p$ in \layeredgraph{B} and let $s_p$ be the pure strategy of player B, equivalent to $p$. We use \weight{k}{i} to denote the expected payoff of player B over battlefield $k$ by putting $i$ troops in it. This means the expected payoff of playing strategy $s_p$ would be 
$\sum \weight{k}{j-i}$
for any $k$, $i$ and $j$ such that there exists an edge from $v_{k, i}$ to $v_{k+1, j}$ in $p$.  It is possible to compute \weight{k}{i} using the following equation:
\begin{equation}
	\weight{k}{i} = \sum\limits_{l=0}^{\troopsA} \pr{k}{l} \times \payoffb_k(i, l) \qquad \forall k:  1\leq t \leq \battlefields
\end{equation}
Note that both equations to compute \pr{k}{i} and \weight{k}{i} are linear and could be computed in our LP.

Assume \weight{k}{i} is the weight of the edge from $v_{k, j}$ to $v_{k+1, i+j}$ in \layeredgraph{B}. Given the probability distribution of player A (which we denoted by \pr{k}{i}), the problem of finding the pure strategy of B with the maximum possible expected payoff, would be equivalent to finding a path from $v_{0, 0}$ to $v_{\battlefields, \troopsB}$ with the maximum weight. 

To find the path with the maximum weight from $v_{0, 0}$ to $v_{\battlefields, \troopsB}$, we define an LP variable \lpvard{k}{i} where its value is equal to the weight of the maximum weighted path from $v_{0, 0}$ to $v_{k, i}$ and we update it using a simple dynamic programming like constraint:
$$\lpvard{k}{i} \geq \lpvard{k-1}{j} + \weight{k-1}{i-j} \qquad \forall i, j: 0 \leq j \leq i \leq \troopsB
$$
 The maximum weighted path from $v_{0, 0}$ to $v_{\battlefields, \troopsB}$ would be equal to the value of $\lpvard{\battlefields}{\troopsB}$. The detailed constraints are shown in Linear Program~\ref{fig:detailedlp}-b.

\begin{lp}
\scalebox{1}{\parbox{1\hsize}{%
\begin{empheq}{alignat*=3}
	\\
	& \hspace{0.7cm} \max \hspace{0.4cm} u\\[0.2cm]
	& \text{(a) }\begin{dcases}
	    \Sigma_{i=0}^{l} \flow_{k,i,l} = \Sigma_{j=l}^{\troopsA} \flow_{k+1,l,j}
			\hspace{0.88cm} \forall k, l:1\leq k \leq \battlefields{}-1, 0 \leq l \leq \troopsA\\
		\flow_{k,i,j} \geq 0 \hspace{3.27cm} \forall k, i, j: 1\leq k \leq \battlefields{}, 0 \leq i\leq j \leq \troopsA \\
		\Sigma_{j=l}^{\troopsA} \flow_{1,l,j} = 0  \hspace{2.6cm} \forall l: 0<l\leq \troopsA \\ 
		\Sigma_{j=0}^{\troopsA} \flow_{1,0,j} = 1 \\
		\Sigma_{j=0}^{\troopsA} \flow_{\battlefields{},j,\troopsA} = 1
	\end{dcases}
	\\[0.3cm]
	& \text{(b) }\begin{dcases}
		\pr{k}{i}= \Sigma_{i=0}^{\troopsA-j} \flow_{k,i,i+j} \hspace{1.7cm} \forall k,j: 1\leq k \leq \battlefields, 0\leq j \leq \troopsA \\
		\weight{k}{i} = \Sigma_{l=0}^{\troopsA} \pr{k}{l} \times \payoffb_k(i, l) \hspace{0.65cm} \forall k, i:  1\leq k \leq \battlefields, 0\leq i \leq \troopsB\\
		\lpvard{0}{i} = 0 \hspace{3.4cm} \forall i: 0 \leq i \leq \troopsB\\
		\lpvard{k}{i} \geq \lpvard{k-1}{j} + \weight{k-1}{i-j} \hspace{0.7cm} \forall i, j:  0 \leq j \leq i \leq \troopsB \\
		\lpvard{\battlefields}{\troopsB} \leq -\gpayoff
	\end{dcases}
	\\
\end{empheq}
}}
\caption{The detailed linear program to find a maxmin strategy for player A. The first set of constraints denoted by (a) ensure we get a valid flow of size $1$ from $v_{0, 0}$ to $v_{\battlefields, \troopsA}$ in the layered graph of player A (a mixed strategy of him) and the second set of constraints denoted by (b) ensure the guaranteed payoff of player A is at least $\gpayoff$. The value of variable $\flow_{k, i, j}$ is the amount of flow passing through the edge from $v_{k, i}$ to $v_{k+1, j}$ for any valid $k$, $i$ and $j$. Variable \lpvard{i}{j} is the size of the maximum weighted path from $v_{0, 0}$ to $v_{i, j}$ in the layered graph of player B, therefore \lpvard{\battlefields}{\troopsB} denotes the maximum payoff of B and $\gpayoff$ is the guaranteed payoff of player A. For an informal explanation of the LP see the text.}
\label{fig:detailedlp}
\end{lp}

Note that the number of variables we use in Linear Program~\ref{fig:detailedlp} is as follows:
\begin{itemize}
	\item Variables of type $\flow_{k, i, l}$: $\Theta(\troopsA^2\battlefields)$.
	\item Variables of type \pr{k}{i} : $\Theta(\troopsA\battlefields)$.
	\item Variables of type \weight{k}{i} : $\Theta(\troopsB\battlefields)$.
	\item Variables of type \lpvard{k}{i} : $\Theta(\troopsB\battlefields)$.
\end{itemize}
Therefore the total number of variables is \orderlowerbound{}. Also note that the number of non-negativity constraints ($\flow_{k, i, j} \geq 0$) is more than any other constraints and is \orderlowerbound{}, therefore the total number of constraints is also \orderlowerbound{}.
\end{proof}

To obtain a mixed strategy for player A, it suffices to run Linear Program~\ref{fig:detailedlp} and find a mixed strategy of A that is mapped to the flow it finds. Note that based on Lemma~\ref{flowtomixLemma} such mixed strategy always exists. Afterwards we do the same for player B by simply substituting \troopsA{} and \troopsB{} in the LP.

%% file: lowerbound.tex
\section{Lower Bound}\label{lowerbound}
A classic approach to reduce the number of LP constraints needed to describe a polytope is to do it in a higher dimension. More precisely, adding extra variables might reduce the number of facets of a polytope. This means a complex polytope may be much simpler in a higher dimension. This is exactly what we did in Section \ref{lp} to improve \ahmadinejadduels's algorithm. In this section we prove that any LP formulation that describes solutions of a Blotto game requires at least \orderlowerbound{} constraints, no matter what the dimension is. This proves the given LP in Section \ref{lp} is tight up to constant factors.

The minimum needed number of constraints in any formulation of a polytope $P$ is called \textit{extension complexity} of $P$, denoted by $\xc{P}$. It is not usually easy to prove a lower bound directly on the extension complexity, because all possible formulations of the polytope must be considered. A very useful technique given by \yannikakisexpressing~\cite{yannakakis1988expressing} is to  prove a lower bound on the \textit{positive rank} of the \textit{slack matrix} of $P$ which is proven to be equal to $\xc{P}$. Note that you could define the slack matrix over any formulation of $P$ and its positive rank would be equal to $\xc{P}$, which means you do not have to worry about all possible formulations. To prove this lower bound we use a method called \textit{rectangle covering lower bound}, already given in \yannikakisexpressing's paper. We will now formally define some of the concepts we used:

\begin{definition}[Extension Complexity]
	Extension complexity of a polytope $P$, denoted by $\xc{P}$ is the smallest number of facets of any other higher dimensional polytope $Q$ that has a linear projection function $\pi$ with $\pi(Q)=P$.
\end{definition}
The next concept we need is slack matrix, which is a matrix of non-negative real values where its columns correspond to vertices of $P$ and its rows correspond to its facets. The value of each element of slack matrix is basically the distance of the vertex corresponding to its column from the facet corresponding to its row. More formally:
\begin{definition}[Slack Matrix]
	Let $\{v_1, \ldots, v_v\}$ be the set of vertices of $P$ and let $\{x \in \mathbb{R}^n | Ax \leq b \}$ be the description of it. The slack matrix of $P$ denoted by $\slack{P}$, is defined by $\slack{P}_{ij} = b_i - A_iv_j$.
\end{definition}
 Also, the non-negative rank of a matrix $S$ is the minimum number $m$ such that $S$ could be factored into two non-negative matrices $F$ and $V$ with dimensions $f\times m$ and $m \times v$.
\begin{definition}[Non-negative Rank]
We define the non-negative rank of a matrix $S$ with $f$ rows and $v$ columns, denoted by $\rk(S)$ to be:  
 \begin{equation}
	\rk (S) = \min\{m | \exists F \in \mathbb{R}^{f \times m}_{\geq 0}, V \in \mathbb{R}^{m \times v}_{\geq 0} : S = FV\}
\end{equation}
\end{definition}
\yannikakisexpressing~\cite{yannakakis1988expressing} proved that $\xc{P} = \rk(\slack{P})$. Therefore instead of proving a lower bound on the extension complexity of $P$, it only suffices to prove a lower bound on the positive rank of the corresponding slack matrix. As mentioned before, to do so, we will use the rectangle covering lower bound. A rectangle covering for a given non-negative matrix $S$ is the minimum number of rectangles needed, to cover all the positive elements of $S$ and none of its zeros (Figure~\ref{fig:rectangleCovering}), formally defined as follows:
\begin{definition}[Rectangle Covering]
	Suppose $r=\rk(S)$ and let $S=UV$ be a factorization of $S$ by non-negative matrices $U$ and $V$. Let $\operatorname{supp}(S)$ denote the set of all the positive values of $S$. Then
	\begin{equation*}
	\operatorname{supp}(S) = \bigcup_{l=1}^{r}(\{i|U_{il}>0\} \times \{j|V_{lj}>0\})
	\end{equation*}
	is a rectangle covering of $S$ with $r$ rectangles.
\end{definition}
\begin{figure}[hbt]
	\centering
	\includegraphics[scale=0.8]{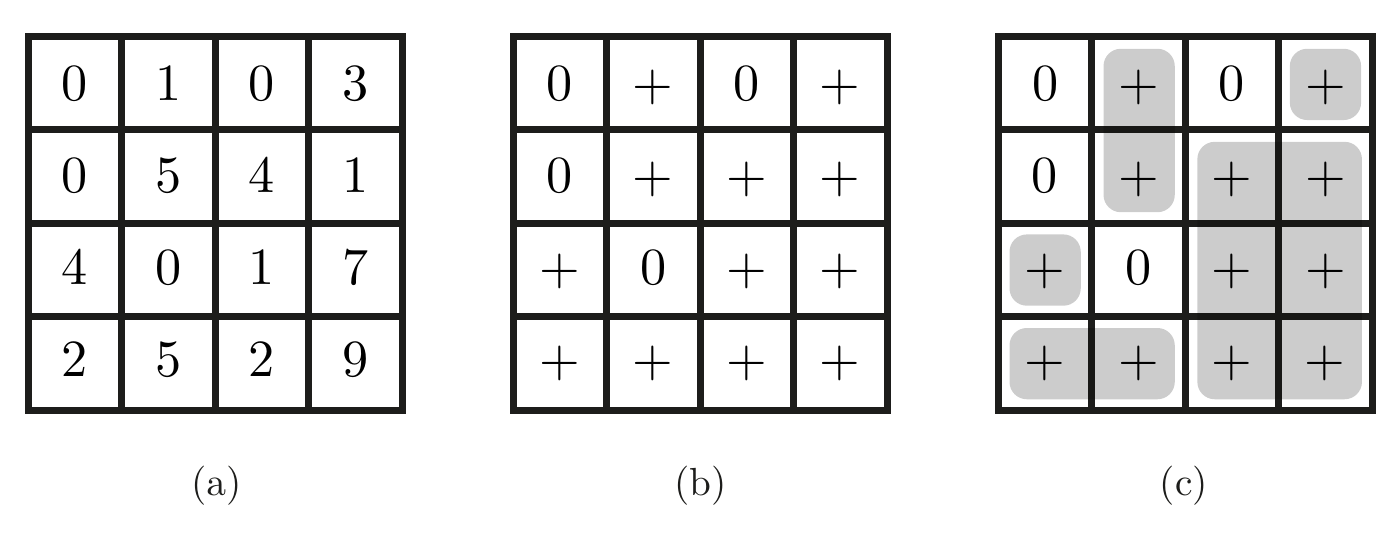}
	\caption{Figure (a) shows a sample matrix, in Figure (b) we change any non-negative value in the matrix of Figure (a) to ``$+$'' and in Figure (c) all these non-negative elements are covered by the minimum possible number of rectangles. Note that the non-negative rank of the matrix in Figure (a) could not be less than $5$ (the number of rectangles).}
	\label{fig:rectangleCovering}
\end{figure}
\yannikakisexpressing{} showed that the number of rectangles in a minimum rectangle covering could never be greater than $\rk(S)$, using a very simple proof. This means any lower bound of it, is also a lower bound of the actual $\rk(S)$. This is the technique we use in the proof of the following lemma, which is used later to prove the main theorem:

\begin{lemma}
\label{lem:minXCMembership}
The extension complexity of the membership polytope of a player in an instance of Blotto with \battlefields{} battlefields and \troops{} troops for each player is at least \orderlowerbound {}.
\end{lemma}

\begin{proof}
	Assume w.l.g. that we are trying to describe the polytope of all valid strategies of player $A$, denoted by $P$. One way of describing this polytope was explained in the LP described in Section \ref{lp}. Now from its membership constraints, only consider the ones that ensure the non-negativity of the flow passing through the edges of the layered graph of player $A$:
	\begin{equation}
	\label{eq:nonnegativity}
		\flow_{i,j,t} \geq 0 \qquad \forall i, j, t: 0\leq i \leq \battlefields-1, 0 \leq j \leq j+t \leq \troopsA
	\end{equation}
	From now on, only consider the part of the slack matrix corresponding to these constraints (we may occasionally call it the slack matrix), its columns as mentioned before, correspond to the vertices of the polytope, which in this case are all possible pure strategies of player $A$. Also its rows correspond to the mentioned constraints. Recall that any pure strategy is a canonical path in the layered graph of player $A$. Note that the slack matrix element corresponding to any arbitrarily chosen non-negativity constraint $e \geq 0$ and any arbitrary vertex $v_j$ corresponding to a pure strategy $S$ is $0$ iff the equivalent canonical path of $S$ does not contain $e$ and is $1$ if it does; since the elements of the slack matrix are calculated using the formula $\slack{P}_{ij} = b - A_iv_j$ and in this case $b$ is always zero and $A_iv_j$ is $-1$ iff $S$ contains the edge in the constraint and is zero otherwise. This implies it is only consisted of zero and one values.
	
	We call any edge $\flow_{b, i, j}$ with $j-i > \frac{\troops}{2}$ a long edge. A canonical path may only contain at most one such edge. On the other hand, any rectangle in the rectangle covering is basically a set of vertices and a set of constraints. Note that all the equivalent pure strategies of those vertices must contain the edges over which the constraints are defined. Therefore no rectangle could contain more than one constraint over long edges. The number of long edges in the layered graph is exactly
	\begin{equation}
		\exactlowerbound{}.
	\end{equation}
	Therefore the minimum number of rectangles to cover all non-negative elements of the slack matrix is at least of the same size and therefore \orderlowerbound{}. 
\end{proof}

\begin{theorem}
\label{XCThm}
	In an instance of Blotto with \battlefields{} battlefields and \troops{} troops for each player the extension complexity of \polytopea{} is \orderlowerbound{}. 
\end{theorem}
\begin{proof}
	Assume the utility function is defined as follows:
	\begin{equation}
	\payoffa(\hat{x}, \hat{y}) = 0 \qquad \forall \hat{x}, \hat{y}.
	\end{equation}
	This means any possible strategy is a maxmin strategy for both players. In particular, the polytope of all possible maxmin strategies of any arbitrarily chosen player of this game, denoted by $P$ contains all possible valid strategies. Now using Lemma~\ref{lem:minXCMembership} we know $\xc{P}$ is at least \orderlowerbound{}. On the other hand, in Section~\ref{lp} we gave an LP with \orderlowerbound{} constraints to formulate the maxmin polytope, therefore the extension complexity of it is exactly \orderlowerbound{}.
\end{proof}

%% file: multidimen.tex
\section{Multi-Resource Colonel Blotto}
In this section we explain how our results could be generalized to solve Multi-Resource Colonel Blotto, or \textit{MRCB}. We define MRCB to be exactly the same game as  Colonel Blotto, except instead of having only one type of resource (troops), players may have any constant number of resource types. Examples of resource types would be time, money, energy, etc.

To solve MRCB we generalize some of the concepts we defined for Colonel Blotto. We first define generalized layered graphs and generalized canonical paths as follows:  
\begin{definition}[Generalized Layered Graph]
	Let $\troops_m$ denote the total number of available resources of $m$-th resource type for player $X$. The generalized layered graph of $X$ has $\battlefields \times \troops{}_1 \times \ldots \times \troops{}_{\resources}$ vertices denoted by $v(i, r_1, \ldots, r_\resources)$, with a directed edge from $v(i, r_1, \ldots, r_{m-1}, x, r_{m+1}, \ldots, r_\resources)$ to $v(i+1, r_1, \ldots, r_{m-1}, y, r_{m+1}, \ldots, r_\resources)$ for any possible $i$, $r$ and $0 \leq x \leq y \leq \troops_m$.
\end{definition}
\begin{definition}[Generalized Canonical Path]
	A generalized canonical path is defined over a generalized layered graph and is a directed path from $v_{0, 0, \ldots, 0}$ to $v_{\battlefields, \troops_1, \ldots, \troops_{\resources}}$.
\end{definition}
By these generalization we can simply prove that pure strategies of a player are equivalent to canonical paths in his generalized layered graph and there could be a surjective mapping from his mixed strategies to flows of size $1$ from $v(0, \ldots, 0)$ to $v(\battlefields, \troops_1, \ldots, \troops_\resources)$ using similar techniques we used in Section~\ref{lp}.
\begin{lemma}
	Each pure strategy for a player in an instance of MRCB is equivalent to exactly one generalized canonical path in the generalized layered graph of him and vice versa.
\end{lemma}
\begin{lemma}
	For any flow $f$ of size $1$ from $v(0, \ldots, 0)$ to $v(\battlefields, \troops_1, \ldots, \troops_\resources)$ in the generalized layered graph of a player with $\troops{}_i$ troops of type $i$, there is at least one mixed strategy with a polynomial size support that is mapped to $f$.
\end{lemma}
Using these properties, we can prove the following theorem:
\begin{theorem}
	In an instance of MRCB, $\polytopeaMRCB{}$ could be formulated with $O(\xcMRCB{})$ constraints and \variablesMRCB{} variables.
\end{theorem}\label{lpMRCB}
\begin{proof}
	The linear program would again look like this:
\begin{alignat}{3}
	\max  \quad & \gpayoff & \\
    \text{s.t. }\qquad
    & \hat{x} \text{ is a valid strategy for player A}\nonumber\\
    & \payoffb(\hat{x},\hat{y}) \leq -\gpayoff \qquad \forall \hat{y}\nonumber
\end{alignat}
	For the first set of constraints (membership constraints) we can use the flow constraints over the generalized layered graph of player A to make sure we have a valid flow of size $1$ from $v(0, \ldots, 0)$ to $v(K, \troops_1, \ldots, \troops_\resources)$. And for the second constraint (payoff constraint) we can find the maximum payoff of player B using a very similar set of constraints to the described one in Section~\ref{lp}, but over the generalized layered graph of player B.
\end{proof}

We can also prove the following lowerbound for MRCB.
\begin{theorem}
	In an instance of MRCB, the extension complexity of $\polytopeaMRCB$ is $\Theta(\xcMRCB)$.
\end{theorem}
\begin{proof}
	The proof is very similar to the proof of Theorem~\ref{XCThm}. We only consider the rectangle covering lower bound over the part of the slack matrix corresponding to the non-negativity of flow through edges in the maxmin. We call an edge from $v(i, r_1, \ldots, r_{m-1}, x, r_{m+1}, \ldots, r_\resources)$ to $v(i+1, r_1, \ldots, r_{m-1}, y, r_{m+1}, \ldots, r_\resources)$ long if $y-x > \frac{n_m}{2}$. No generalzied canonical path could contain more than $\resources$ long edges therefore no rectangle could cover more than $\resources$ constraints. On the other hand there are $\Theta(\xcMRCB{})$ long edges in the layered graph. Since $\resources$ is a constant number the extension complexity is $\Omega(\xcMRCB)$. Moreover since we already gave a possible formulation with $O(\xcMRCB)$ constraints in Theorem~\ref{lpMRCB} the extension complexity is also $O(\xcMRCB)$ and therefore $\Theta(\xcMRCB)$. 
\end{proof}

%% file: experimental.tex
\section{Experimental Results}
We implemented the algorithm described in Section~\ref{lp} using Simplex method to solve the LP. We ran the code on a machine with a dual-core processor and an 8GB memory. The running time and the number of constraints of the LP for each input is shown in Table~\ref{tbl:runtime}.
\begin{table}[]
\centering
\begin{tabular}{|c|c|c|c|c|}
\hline
$\battlefields$ & $\troopsA$ & $\troopsB$ & Constraints & Running Time \\
\hline
10 & 20 & 20 & 3595 & 0m3.575s\\
10 & 20 & 25 & 4855 & 0m3.993s\\
10 & 20 & 30 & 6365 & 0m6.695s\\
10 & 25 & 25 & 5295 & 0m8.245s\\
10 & 25 & 30 & 6805 & 0m7.502s\\
10 & 30 & 30 & 7320 & 0m30.955s\\
15 & 20 & 20 & 5065 & 0m14.965s\\
15 & 20 & 25 & 6950 & 0m11.842s\\
15 & 20 & 30 & 9210 & 0m24.196s\\
15 & 25 & 25 & 7440 & 0m46.165s\\
15 & 25 & 30 & 9700 & 0m31.714s\\
15 & 30 & 30 & 10265 & 2m20.776s\\
20 & 20 & 20 & 6535 & 0m46.282s\\
20 & 20 & 25 & 9045 & 0m35.758s\\
20 & 20 & 30 & 12055 & 0m38.507s\\
20 & 25 & 25 & 9585 & 1m38.367s\\
20 & 25 & 30 & 12595 & 0m51.795s\\
20 & 30 & 30 & 13210 & 9m13.288s\\
\hline
\end{tabular}
\caption{The number of constraints and the running time of the implemented Colonel Blotto based on different inputs. The first column shows the number of battlefields, the second and third columns show the number of troops of player A and B respectively. The number of constraints does not include the non-negativity constraints since by default every variable was assumed to be non-negative in the library we used.
}
\label{tbl:runtime}
\end{table}
Using this fast implementation we were able to run the code for different cases. In this section we will mention some of our observations that mostly confirm the theoretical predications.

An instance of Colonel Blotto is symmetric if the payoff function is the same for all battlefields, or in other words for any pure strategies $x$ and $y$ for player A and B and for any two battlefields $i$ and $j$, $\payoffa_i(x, y) = \payoffa_j(x, y)$. Also, an instance of blotto is auctionary if the player allocating more troops in a battlefield wins it (gets more payoff over that battlefield). More formally in an auctionary instance of Colonel Blotto, if $x$ and $y$ are some pure strategies for player A and B respectively, then
\begin{equation*}
	    \payoffa_i(x, y)= 
\begin{dcases}
    +w(i), & \text{if } x_i > y_i\\
    0, & \text{if } x_i = y_i\\
    -w(i), & \text{otherwise}
\end{dcases}
\end{equation*}
Recall that $x_i$ and $y_i$ denote the amount of troops A and B put in the $i$-th battlefield respectively.

Note that in an auctionary Colonel Blotto if $\troopsA \geq (\troopsB+1)\battlefields$, then by putting $\troopsB+1$ troops in each battlefield, player A wins all the battlefields and gets the maximum possible overall payoff. On the other hand if $\troopsA=\troopsB$, the payoff of player A in any Nash equilibrium is exactly $0$ because there is no difference between player A and player B by definition of an auctionary Colonel Blotto if $\troopsA=\troopsB$, and any strategy for A could also be used for B and vice versa. W.l.g. we can ignore the case where $\troopsA < \troopsB$. However, it is not easy to guess the payoff of A in a Nash equilibrium if $\troopsB \leq \troopsA < (\troopsB+1)\battlefields$. After running the code for different inputs, we noticed the growth of $\payoffa{}$ with respect to $\troopsA$ (when $\troopsB$ is fixed) has a common shape for all inputs. Figure~\ref{fig:charts} shows the chart for different values of $\troopsA$, $\troopsB$ and $\troops$.


\begin{figure}
	\centering
	\includegraphics[scale=0.6]{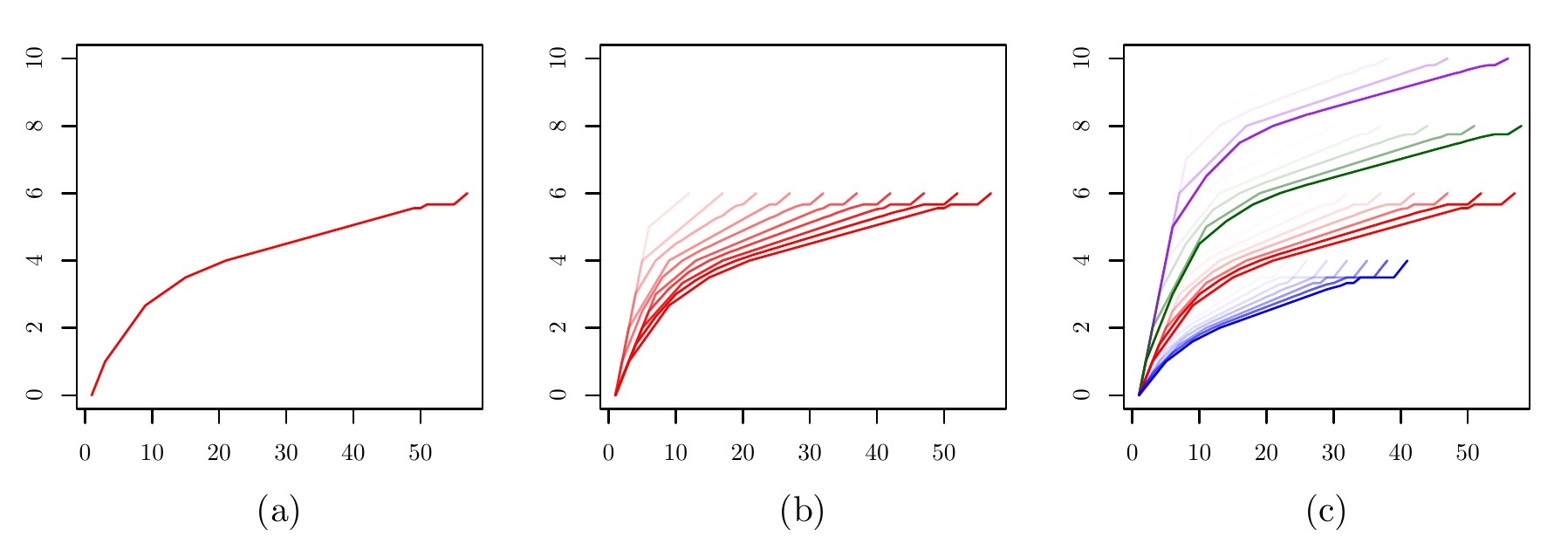}
	\caption{The y-axis is the payoff of A in the Nash equilibrium and the x-axis shows the value of $\troopsA - \troopsB$. In Figure (a), $\battlefields=6$ and $\troopsB=10$. In Figure (b), $\battlefields=6$ and for different values of $\troopsB$ in the range of 1 to 10 the same diagram as Figure (a-) is drawn. Figure (c) is the same plot as Figure (b) but for different values of $\battlefields$. For instance for the blue lines $\battlefields=4$, for the red lines $\battlefields=6$, for the green lines $\battlefields=8$ and for the purple lines $\battlefields=10$. In all examples payoff function of player A over a battlefield $i$, is $sgn(x_i-y_i)$ where $x_i$ and $y_i$ denote the number of troops A and B put in the $i$-th battlefield respectively.}
	\label{fig:charts}
\end{figure}

There has been several attempts to mathematically find the optimum payoff of players under different conditions. For example Roberson \cite{R06} considered the continuous version of Colonel Blotto and solved it. Hart \cite{H07} solved the symmetric and auctionary model and solved it for some special cases. Little is known about whether it is possible to completely solve the discrete version when the game is symmetric and auctionary or not.

Surprisingly, we observed the payoff of players in the symmetric and auctionary discrete version, is very close to the continuous version Roberson considered. The payoffs are specially very close when the number of troops are large compared to the number of battlefields, making the strategies more flexible and more similar to the continuous version. Figure~\ref{fig:roberson} compares the payoffs in the aforementioned models. In Roberson's model in case of a tie, the player with more resources wins while in the normal case there is no such assumption; however a tie rarely happens since by adding any small amount of resources the player losing the battlefield would win it.
\begin{figure}
	\centering
	\includegraphics[scale=0.6]{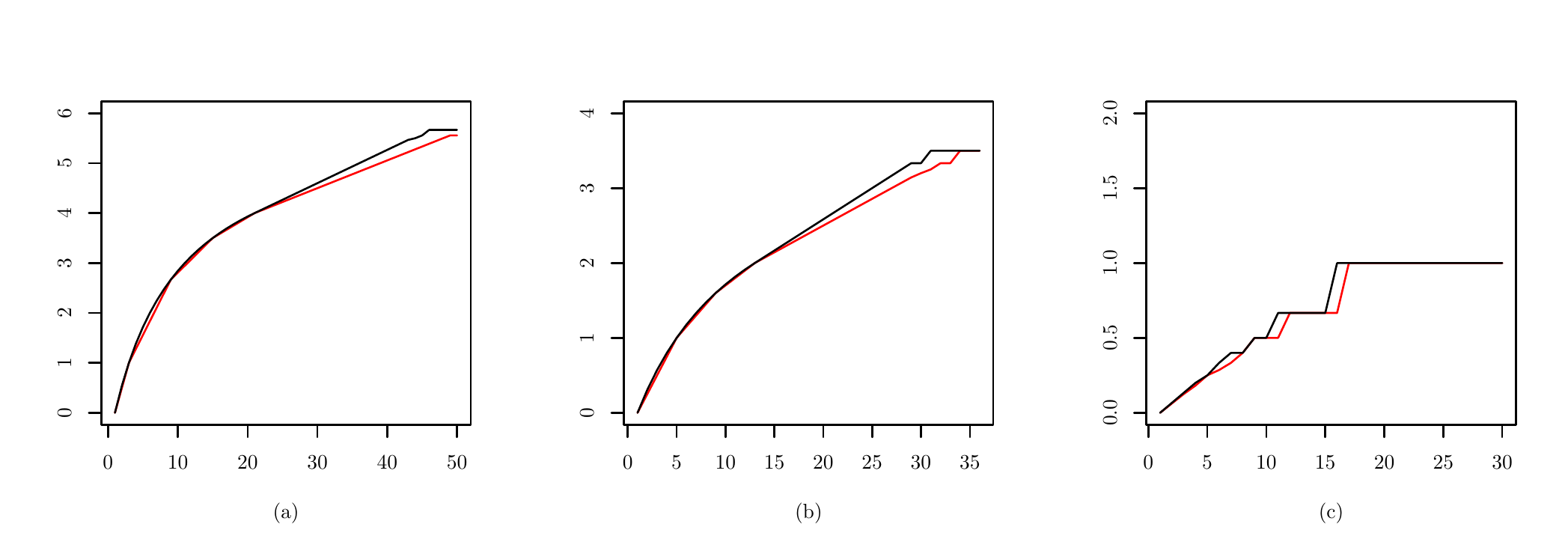}
	\caption{The y-axis is the payoff of A in the Nash equilibrium and the x-axis shows the value of $\troopsA - \troopsB$. The black and red line show the payoff in the continuous model and discrete model respectively. In figure (a), $\battlefields=6$ and $\troopsB=10$, in figure (b), $\battlefields=4$ and $\troopsB=12$ and in figure (c), $\battlefields=2$ and $\troopsB=30$.}
	\label{fig:roberson}
\end{figure}